\documentclass[conference]{IEEEtran}

\usepackage{enumerate}
\usepackage{stmaryrd}
\usepackage{amsmath,amssymb,amsthm}
\usepackage[colorlinks,citecolor=red]{hyperref}
\usepackage{tikz}
\usetikzlibrary{arrows,automata,backgrounds}
\usepgflibrary{shapes}
\usepackage{OSSLC-macros}

\usepackage{booktabs}  
\usepackage{subcaption} 

\usepackage[switch,mathlines]{lineno}

\begin{document}

\newtheorem{theorem}{Theorem}
\newtheorem{lemma}[theorem]{Lemma}
\newtheorem{corollary}[theorem]{Corollary}
\newtheorem{definition}[theorem]{Definition}

\title{Universal Semantics for the Stochastic $\lambda$-Calculus}     %

\author{
\IEEEauthorblockN{ Pedro H. Azevedo de Amorim\IEEEauthorrefmark{1}, Dexter Kozen\IEEEauthorrefmark{1}, Radu Mardare\IEEEauthorrefmark{2}, Prakash Panangaden\IEEEauthorrefmark{3}, Michael Roberts\IEEEauthorrefmark{1}}
\IEEEauthorblockA{\IEEEauthorrefmark{1}Cornell University}
\IEEEauthorblockA{\IEEEauthorrefmark{2}University of Strathclyde}
\IEEEauthorblockA{\IEEEauthorrefmark{3}McGill University}
}

\IEEEoverridecommandlockouts
\IEEEpubid{\makebox[\columnwidth]{978-1-6654-4895-6/21/\$31.00~
\copyright2021 IEEE \hfill} \hspace{\columnsep}\makebox[\columnwidth]{ }}

\maketitle

\begin{abstract}
We define sound and adequate denotational and operational semantics for the stochastic lambda calculus. These two semantic approaches build on previous work that used an explicit source of randomness to reason about higher-order probabilistic programs.
\end{abstract}

\section{Introduction}


Probabilistic programming has enjoyed a recent resurgence of interest, buoyed by the emergence of new languages and applications in the statistical analysis of large datasets and machine learning. Recent foundational research has focused on semantic models for higher-order functional languages.

%

One approach that is radically different from other approaches is that of
\cite{BFKMPS18a}, which involves Boolean-valued models for the stochastic
$\lambda$-calculus. Based on an original idea of Scott~\cite{Scott14}, the
paper \cite{BFKMPS18a} succeeded in incorporating random variables in a
set-theoretic model of the untyped $\lambda$-calculus. The approach was
formulated in terms of a nonstandard Boolean-valued interpretation of set theory
based on the idea of Boolean-valued models of ZF set theory
(see~\cite{Bell2011}). Boolean-valued models were first introduced by Scott
\cite{Scott67} as an alternative technique to Cohen forcing for obtaining
independence results in set theory. The independence of the Continuum
Hypothesis was obtained by introducing an arbitrarily large set of real-valued
random variables. The measure algebra of a standard Borel space $\Omega$, a
complete Boolean algebra, was used as a set of generalized truth values instead
of the usual two-element Boolean algebra.

Scott also observed that these ideas could be given a probabilistic
interpretation. The basic intuitions were briefly laid out in \cite{Scott14}
and the formal development carried out in \cite{BFKMPS18a}. The primary goal
was to develop an equational theory in which equations between stochastic
$\lambda$-terms have probabilistic meaning and take values in a complete Boolean
algebra. The intention was to provide reasoning principles for evaluating the
equality of $\lambda$-terms under various program transformations.

The language contains a binary probabilistic choice operator $\oplus$, which
captures the idea that a choice is to be made between two terms based on a
random process. The source of randomness is called a \emph{tossing process}, a
random variable $T:\Omega\to\Bin$ giving a sequence of independent fair coin
flips.

The semantics presented in this paper differ from those of \cite{BFKMPS18a} in several key ways. The semantics of \cite{BFKMPS18a} use static scoping for random coins. This causes $\beta$-reduction for unrestricted terms to be unsound, as 
a random coin may be used in more than one probabilistic decision. It is sound only
under a certain restriction, namely that all probabilistic decisions in the
argument be resolved before applying the function. This is a major impediment to the development of an operational semantics for which adequacy can be proved; indeed an operational semantics is not given in \cite{BFKMPS18a}. In contrast, we dynamically scope random coins, allowing them to be supplied at function call time. The nonstandard Boolean-valued foundations of the \cite{BFKMPS18a} semantics further complicate the development of an operational semantics, as they would call for a Boolean-valued operational semantics. In this work, we present a semantics with simpler domain-theoretic semantics with standard foundations.

A more operational approach was taken in \cite{borgstrom}. That work presented
an operational semantics for the stochastic $\lambda$-calculus as an idealized
version of the Church language \cite{goodman2012church}, along with reasoning
principles and applications to the correctness of an implementation of trace
Markov chain Monte Carlo processes. That work did not define a denotational
semantics, which obliged them to reason combinatorially about programs.

In this paper we modify the approaches described above to conform to each
other. We amend the stochastic denotational semantics of \cite{BFKMPS18a} to
alter the scoping discipline of random sources in a way that still permits the
Boolean-valued view of \cite{BFKMPS18a}, yet allows the formulation of big- and
small-step operational rules similar to \cite{borgstrom} without the artificial
restriction mentioned above. We prove soundness and adequacy of the operational
semantics with respect to the reformulated stochastic semantics of
\cite{BFKMPS18a}, solving the main problem left open in that paper.

The organization of this paper and our main contributions are as follows.

\subsubsection*{Syntax}

In \S\ref{sec:syntax} we review the syntax of the stochastic $\lambda$-calculus as presented in \cite{BFKMPS18a}, but with one change: We use capsules to represent recursive functions instead of an explicit fixpoint constructor. A \emph{capsule} \cite{JK12b} is a pair $\angle{M,\sigma}$, where $M$ is a stochastic $\lambda$-term and $\sigma$ is an environment, such that
\begin{itemize}
	\item
	$\FV(M)\subs\dom\sigma$, and 
	\item
	$\forall x\in\dom\sigma\ \FV(\sigma(x))\subs\dom\sigma$.
\end{itemize} 
Capsules represent a finite coalgebraic representation of a closed regular $\lambda$-coterm (an infinite $\lambda$-term). This representation obviates the need for an explicit fixpoint constructor.

\subsubsection*{Tossing Processes}

In \S\ref{sec:tossing} we undertake a comprehensive exposition of \emph{tossing
 processes}, or measure-preserving transformations of the Cantor space of
infinite coin sequences. These processes arise in the study of behavioral
invariance of programs, i.e.~programs that behave the same way except for coin
usage. We characterize the computable and continuous processes, both partial and
total, and show their relationship to prefix codes. We also identify a general
class of processes called \emph{tree processes} that we later use in
\S\ref{sec:opsem} to characterize the relationship between the coin usage
patterns of our big- and small-step operational semantics.

\subsubsection*{Computability of tossing processes}

Also in \S\ref{sec:tossing}, we show how to embed the Cantor space $\Bin$ in a
Scott domain in a natural way, thereby laying the groundwork for our modified
denotational semantics. Consider the set $\Binp$ of finite and infinite binary
strings ordered by the prefix relation. This is an algebraic DCPO whose compact
elements are the finite strings. Define $\up x = \set{y\in\Binp}{x\preceq y}$
for $x\in\Binp$, where $\preceq$ is the prefix relation. The basic Scott-open
sets are $\up x$ for $x\in 2^*$. These are well known folklore
results;\footnote{\url{https://en.wikipedia.org/wiki/Scott_domain}} the domain is usually known as the domain of binary
streams.  

The infinite streams or sequences, with the subspace topology inherited from the
Scott topology, is homeomorphic to Cantor space. Lemmas \ref{lem:Scottify1} and
\ref{lem:Scottify2} establish a formal relationship between these two spaces and
their continuous maps. This ``Scottified'' Cantor space gives an explicit
characterization of functions that behave continuously with respect to coin
usage in the sense that halting computations depend only on finite prefixes of
the coin sequence. This allows us to discuss computable and continuous tossing
processes. All the tree processes are Scott-continuous.

\subsubsection*{A Simplified Stochastic Semantics}

In \S\ref{sec:stochastic}, we review the stochastic denotational semantics of
\cite{BFKMPS18a}. 
That semantics is based on a semantic map
\begin{align*}
\psem-:\Exp\to\Env\to\Cont\to\Toss\to\RV
\end{align*}
where
\begin{itemize}
	\item
	$\RV$ is the set of \emph{random variables} $\Omega\to\Val$ from a sample space $\Omega$ taking values in a reflexive CPO $\Val$,
	\item
	$\Exp$ is the set of \emph{stochastic $\lambda$-terms} $M$,
	\item
	$\Env$ is the set of \emph{environments} $e:\Var\to\RV$,
	\item
	$\Cont$ is the set of \emph{continuations} $c:\RV\to\RV$, and
	\item
	$\Toss$ is the set of \emph{tossing processes} $T:\Omega\to\Bin$.
\end{itemize}

Thus $\psem MECT:\RV$.

We can simplify the exposition as follows:
\begin{itemize}
	\item
	Suppose we restrict continuations to be of the form $SF=\lamb f{\lamb\omega{F\omega(f\omega)}}$ for some $F:\Omega\to[\Val\to\Val]$, where $[\Val\to\Val]$ denotes the Scott-continuous deterministic maps.\footnote{The operation $S$ is the familiar $S$-combinator from combinatory logic.} Then all continuations that arise in the inductive definition of $\psem M$ are also of this form. Formally adopting this restriction allows us to eliminate continuations altogether.
	\item
	A tossing process determines how a supplied source of randomness is used in a computation. In \cite{BFKMPS18a} they are of type $\Omega\to\Bin$, where $\Omega$ is an abstract sample space. For our purposes, there is no reason not to assume that the sample space is $\Bin$ with Lebesgue measure, so a \emph{tossing process} is now any measurable map $T:\Bin\to\Bin$ such that $T^{-1}$ preserves measure. Examples are $\tl(\alpha)=\alpha_1\alpha_2\alpha_3\cdots$ and $\evens(\alpha)=\alpha_0\alpha_2\alpha_4\cdots$~. This allows a more concrete treatment as developed in \S\ref{sec:tossing}.
	\item
	We can omit the fixpoint operator of \cite{BFKMPS18a} using capsules \cite{JK12b}, as described in \S\ref{sec:tossing}.
\end{itemize}

In the treatment of \cite{BFKMPS18a}, general $\beta$-reduction is unsound, precluding any standard operational semantics. This is because the source of randomness used by a function in the evaluation of its body is a coin sequence packaged with the function at the site of the function's definition. Thus randomness, like environments, is statically scoped. This can lead to the reuse of coins at different locations in the program, thereby breaking linearity. For example, in the evaluation of $(\lamb x{xa(xb)})(\lamb y{cy\oplus dy})$, the same coin is used twice in the resolution of two $\oplus$'s when the body of the first expression is evaluated.

To achieve adequacy with respect to an operational semantics, we modify the denotational semantics of functions to allow the random source to be supplied as a parameter at the call site.

\subsubsection*{Deterministic Denotational Semantics}

In \S\ref{sec:deterministic}, we observe that in the stochastic semantics, the value of $\psem{-}$ depends not on the whole tossing process $T$ nor the environment $E$, which are random variables parameterized by a sample point $\omega\in\Omega$, but only on their values. Intuitively, each run of the program corresponds to one trial, which is determined by a single sample point $\omega$. This is the same observation used to eliminate continuations. This allows us to develop an intermediate \emph{deterministic} denotational semantics in which probabilistic choices are resolved in advance, after which the program runs deterministically, making probabilistic decisions based on a presampled infinite stack of random numbers.

The deterministic denotational semantics is built on a reflexive domain of values constructed using the Scottified Cantor space of \S\ref{sec:tossing}. In \S\ref{sec:equivalence}, we prove the equivalence of the stochastic denotational semantics of \cite{BFKMPS18a} (as modified in \S\ref{sec:stochastic}) and the deterministic semantics of \S\ref{sec:deterministic} (Theorem \ref{thm:equivalence}).

\subsubsection*{Operational Rules}

In \S\ref{sec:opsem}, we give big-step and small-step structured operational semantics in the style of \cite{Plotkin81}.
The big-step rules take the form $\BOp Me\alpha vf$, which means that $\cps Me$ reduces to normal form $\cps vf$ with coins $\alpha\in\Bin$. The small-step rules take the form $\cps Me \red x \cps Nf$, which means that $\cps Me$ reduces to $\cps Nf$ via a computation that consumes exactly a prefix $x\in 2^*$ of the infinite coin sequence.

In Theorem \ref{thm:tree}, we prove the equivalence of the big- and small-step rules, which use their random coins in a different pattern. The relationship is characterized by a tree process as described in \S\ref{sec:tossing}.

\subsubsection*{Soundness and Adequacy}

In \S\ref{sec:adequacy}, we prove the soundness and adequacy of our denotational semantics with respect to our big-step operational semantics (Theorem \ref{thm:adequacy}). Unlike most adequacy proofs that use logical relations, this proof is a relatively straightforward inductive argument, as the deterministic denotational semantics and the big-step operational semantics use their coins in the same pattern.

\section{Syntax}
\label{sec:syntax}

Let $\Var$ be a countable set of program variables $x,y,\ldots$~. Let $\Exp$ denote the set of untyped $\lambda$-terms $M,N,K,\ldots$ with the usual abstraction and application operators plus an additional binary operator $\oplus$ for probabilistic choice. Let $\Lambda$ denote the set of $\lambda$-abstractions, $\lambda$-terms of the form $\lamb xM$.

\subsection{Capsules}
\label{sec:capsules}

A \emph{capsule} is a pair $\cps M\sigma$, where $M\in\Exp$ and $\sigma:\Var\rightharpoondown\Lambda$ is a \emph{capsule environment}, such that
\begin{enumerate}[(i)]
	\item
	$\FV(M)\subs\dom\sigma$
	\item
	$\forall x\in\dom\sigma\ \FV(\sigma(x))\subs\dom\sigma$.
\end{enumerate}
Here $\dom\sigma$ refers to the domain of $\sigma$ and $\FV M$ refers to the set of free variables of $M$.
A capsule is \emph{reduced} if its first component is in $\Lambda$. Reduced capsules are denoted with lowercase letters, as $\cps v\sigma$.

A capsule is a finite coalgebraic representation of a regular closed $\lambda$-coterm (infinitary $\lambda$-term), which is an element of the final coalgebra for the signature of the $\lambda$-calculus. Capsules give a convenient representation of recursive functions without the need of fixpoint combinators.

Capsules are considered equivalent modulo $\alpha$-conversion, including $\alpha$-conversion of the variables used in $\sigma$. In terms of nominal sets with the variables as atoms, the support of a capsule is $\emptyset$. Capsules are also considered equivalent modulo garbage collection in the sense that we can assume without loss of generality that $\dom e$ is a minimal set of variables satisfying (i) and (ii).

The capsule $\beta$-reduction rule is
\begin{align*}
& \cps{(\lamb xM)\,v}\sigma \to \cps{M\subst yx}{\sigma\rebind vy}\quad \text{($y$ fresh)}
\end{align*}
applied in a call-by-value evaluation order. This mechanism captures static scoping without closures, heaps, or stacks \cite{JK12b}. Here we are using the notation $\subst--$ for both substitution (as in $M\subst yx$) and rebinding (as in $\sigma\rebind vy$).

Capsules were introduced in \cite{JK12b}. For the stochastic $\lambda$-calculus, we augment the system with the new syntactic construct $M\oplus N$ for probabilistic choice.

\section{Tossing Processes}
\label{sec:tossing}

The \emph{Cantor space} $\Bin$ is the space of infinite bitstreams. It is the topological power of $\omega$ copies of the two-element discrete space $2=\{0,1\}$. Elements of $\Bin$ are denoted $\alpha,\beta,\ldots$~. The topology is generated by basic open sets $I_x = \set{\alpha\in\Bin}{x\prec\alpha}$, where $x\in 2^*$ and $\prec$ denotes the strict prefix relation. The sets $I_x$ are called \emph{intervals}. The topology is also generated by the standard metric $d(\alpha,\beta)=2^{-n}$, where $n$ is the length of the longest common prefix of $\alpha$ and $\beta$, or $0$ if $\alpha=\beta$.

The Borel sets $\B$ of the Cantor space are the smallest $\sigma$-algebra containing the open sets.
The uniform (Lebesgue) measure $\Pr$ on $(\Bin,\B)$ is generated by its values on intervals:
$\Pr(\set\alpha{x\prec\alpha}) = 2^{-\len x}$. The Lebesgue measurable sets are the smallest $\sigma$-algebra containing the Borel sets and all subsets of null sets. The set of null sets is denoted $\N$.

A \emph{tossing process} is any measurable map $T:\Bin\to\Bin$ such that $T^{-1}$ preserves measure; that is, for all $A\in\B$, $\Pr(T^{-1}(A))=\Pr(A)$. Given an infinite bitstream $\alpha = \alpha_0\alpha_1\alpha_2\cdots$, we can define the examples $\tl(\alpha)=\alpha_1\alpha_2\alpha_3\cdots$ and $\evens(\alpha)=\alpha_0\alpha_2\alpha_4\cdots$~. A tossing process determines how a supplied source of randomness is used in a computation.

\begin{lemma}
	\label{lem:measurepreserving}
	$T$ is a tossing process iff for all $x\in 2^*$,
	\begin{align*}
	\Pr(\set\alpha{x\prec T(\alpha)}) &= 2^{-\len x}
	\end{align*}
\end{lemma}
\begin{proof}
	We have $x\prec T(\alpha)$ iff $\alpha\in T^{-1}(\set\gamma{x\prec\gamma})$, therefore
	\begin{align*}
	\lefteqn{\text{$T$ is a tossing process}}\\
	&\Iff \forall x\in 2^*\ \Pr(T^{-1}(\set{\gamma}{x\prec\gamma})) = \Pr(\set{\gamma}{x\prec\gamma})\\
	&\Iff \forall x\in 2^*\ \Pr(\set\alpha{x\prec T(\alpha)}) = 2^{-\len x}.
	\end{align*}
\end{proof}

\subsection{Computable and Continuous Processes}

For a function $f:\Bin\to\Bin$ to be computable, it must be possible to emit each digit of the output stream after reading only finitely many digits of the input stream. For example, one can emit the $n$th digit of $\tl\alpha$ after reading $n+1$ digits of $\alpha$, and one can emit the $n$th digit of $\evens\alpha$ after reading the first $2n-1$ digits of $\alpha$.

\begin{lemma}
	\label{lem:cont1}
	All computable tossing processes $T:\Bin\to\Bin$ are continuous. All continuous functions $\Bin\to\Bin$ are uniformly continuous with respect to the standard metric.
\end{lemma}
\begin{proof}
	To be computable, it must be the case that any finite prefix $x\prec T\alpha$ of the output is determined by some finite prefix of the input $\alpha$. This implies that $x\prec T\beta$ for any $\beta$ that agrees with $\alpha$ on a sufficiently long prefix; in other words, $\set\beta{y\prec\beta}\subs T^{-1}(\set\gamma{x\prec\gamma})$ for some $y\prec\alpha$. Thus $T^{-1}(\set\gamma{x\prec\gamma})$ is open. As $x$ was arbitrary, $T$ is continuous.
	
	It is a standard result that any continuous function on a compact metric space is uniformly continuous.
\end{proof}
For example, $\tl$ is Lipschitz with constant $2$: $d(\tl\alpha,\tl\beta) \le 2d(\alpha,\beta)$. The maps $\evens$ and $\odds$ are not Lipschitz, but they are H\"older of order $1/2$; that is, both maps satisfy $d(T\alpha,T\beta) \le \sqrt{d(\alpha,\beta)}$.

There is a subtle distinction between ``reading'' and ``consuming'' a digit. The latter refers to using the digit to make a probabilistic choice. One can read digits without consuming them; they can be saved to make probabilistic choices later, at which point they are consumed. It is important for independence that digits not be consumed more than once.

Uniform continuity fails if we allow tossing processes to be partial. A \emph{partial tossing process} is a measure-preserving partial measurable function $T:\Bin\pfun\Bin$. Such a function is necessarily almost everywhere defined, since $\dom T = T^{-1}(\Bin)$, which must have measure 1.

\begin{lemma}
	\label{lem:cont2}
	All computable partial tossing processes $T:\Bin\to\Bin$ are continuous.
	There is a computable partial tossing process that is continuous but not uniformly continuous.
\end{lemma}
\begin{proof}
	Computable partial tossing processes $T:\Bin\to\Bin$ are continuous for the same reason that total ones are.
	
	For the second statement, define $T$ coinductively as follows:
	\begin{align*}
	T(0^*10\alpha) &= 0\,T(\alpha) & T(0^*11\alpha) &= 1\,T(\alpha).
	\end{align*}
	The domain of definition of $T$ is $(0^*1)^\omega$, the measure-1 set of streams containing infinitely many 1's. It is continuous, since if $\alpha$ and $\beta$ share a prefix with at least $2n$ 1's, then $T\alpha$ and $T\beta$ share a prefix of length $n$. It is not uniformly continuous, as there is no bound on the number of input digits that need to be read before producing the next output digit.
\end{proof}

The $T$ of the previous lemma is undefined on the nullset $2^*0^\omega$. One can define $T$ arbitrarily on this set, but Lemma \ref{lem:cont1} says that the resulting total tossing process cannot be continuous. This does not rule out the possibility that every total tossing process might be equivalent modulo $\N$ to some continuous partial tossing process. However, this too is false.

\begin{lemma}
	\label{lem:nctp}
	There is a tossing process that is not equivalent modulo $\N$ to any continuous partial tossing process.
\end{lemma}
\begin{proof}
	The proof uses \cite[Exercises 7 and 8, p.~59]{Rudin74}. A complete proof can be found in the Appendix which can be found in the complete version of this paper.
\end{proof}

\begin{lemma}
	\label{lem:surjective}
	All continuous tossing processes, partial or total, are surjective. 
\end{lemma}
\begin{proof}
	For any $\beta\in\Bin$, $T^{-1}(\{\beta\}) = \bigcap_n T^{-1}(\set{\alpha}{\alpha_n=\beta_n})$. This is the intersection of a collection of closed sets with the finite intersection property in a compact space, therefore it is nonempty.
\end{proof}

Every tossing process $T$ is equivalent modulo $\N$ to a partial $T'$ that is ``almost continuous'' in the sense that all ${T'}^{-1}(\set\alpha{x\prec\alpha})$ are $\Delta^0_2$, that is, both $G_\delta$ and $F_\sigma$. One can obtain $T'$ from $T$ by deleting countably many nullsets $G_x\setminus F_x$, where $G_x$ and $F_x$ are $G_\delta$ and $F_\sigma$, respectively, such that $\Pr(F_x)=\Pr(G_x)$ and $F_x\subs T^{-1}(\set\alpha{x\prec\alpha})\subs G_x$. The sets $F_x$ and $G_x$ exist by Lebesgue measurability. 

The following theorem gives a characterization of the continuous partial and total tossing processes $T:\Bin\to\Bin$.
A \emph{binary prefix code} is a nonempty set of prefix-incomparable finite-length binary strings. A binary prefix code $P$ is \emph{exhaustive} if all $\alpha\in\Bin$ have a prefix in $P$. An exhaustive prefix code is necessarily finite by compactness. 

If $P$ and $Q$ are two binary prefix codes, write $P\preceq Q$ if every element of $Q$ is an extension of some element of $P$; that is, for every $y\in Q$, there exists $x\in P$ such that $x\preceq y$.

A \emph{coding function} is a map $x\mapsto P_x$, where $P_x$ is a prefix code, such that
\begin{itemize}
	\item
	$P_\eps = \{\eps\}$;
	\item
	if $x$ and $y$ are prefix-incomparable, then $P_x\cap P_y=\emptyset$;
	\item
	if $x\preceq y$, then $P_x\preceq P_y$.
\end{itemize}
In addition, $x\mapsto P_x$ is said to be \emph{exhaustive} provided
\begin{itemize}
	\item
	if $P$ is an exhaustive prefix code, then so is $\bigcup_{x\in P} P_x$.
\end{itemize}
\begin{theorem}
\label{thm:conttp}
	For every continuous partial tossing process $T$, there is a unique coding function $x\mapsto P_x$ such that
	\begin{enumerate}[{\upshape(i)}]
		\item
		$x\prec T\alpha$ iff $y\prec\alpha$ for some $y\in P_x$; in other words,
		\begin{align*}
		T^{-1}(I_x) = \set\alpha{x\prec T\alpha} = \bigcup_{y\in P_x} I_y;
		\end{align*}
		\item
		$P_x$ is $\preceq$-minimal among prefix codes satisfying {\upshape(i)};
		\item
		$\Pr(\bigcup_{y\in P_x} I_y)=2^{-\len x}$.
	\end{enumerate}
	If $T$ is total, then $x\mapsto P_x$ is exhaustive.
	Moreover, every coding function of this form gives rise to a continuous partial or total tossing process.
\end{theorem}
\begin{proof}
A proof can be found in the Appendix which can be found in the complete version of this paper.
\end{proof}

\subsection{Tree Processes}
\label{sec:trees}

Let $t:2^*\to\omega$ be a labeled tree with no repetition of labels along any path; that is, if $x,y\in 2^*$ with $x\prec y$, then $t(x)\ne t(y)$. Each such tree gives rise to a continuous tossing process $T:\Bin\to\Bin$ as follows. Given $\alpha=\alpha_0\alpha_1\alpha_2\cdots\in\Bin$, let $T(\alpha)=\beta_0\beta_1\beta_2\cdots$, where $\beta_n=\alpha_{t(\beta_0\beta_1\cdots\beta_{n-1})}$. Thus the bit of the input sequence $\alpha$ that is tested in the $n$th step can depend on the outcomes of previous tests as determined by $t$. The restriction ``no repetition of labels along any path'' ensures that no coin is used more than once.

Every such $T$ is measurable and measure-preserving, thus a tossing process:
\begin{align*}
& T^{-1}(\set{\gamma}{\beta_0\cdots\beta_{n-1} \prec \gamma})
= \set{\alpha}{\beta_0\cdots\beta_{n-1} \prec T(\alpha)}\\
&= \set{\alpha}{\bigwedge_{i=0}^{n-1} T(\alpha)_i=\beta_i}
= \bigcap_{i=0}^{n-1} \set{\alpha}{\alpha_{t(\beta_0\cdots\beta_{i-1})}=\beta_i}
\end{align*}
\begin{align*}
\lefteqn{\Pr(T^{-1}(\set{\gamma}{\beta_0\cdots\beta_{n-1} \prec \gamma}))}\\
&= \Pr(\bigcap_{i=0}^{n-1} \set{\alpha}{\alpha_{t(\beta_0\cdots\beta_{i-1})}=\beta_i})\\
&= \prod_{i=0}^{n-1} \Pr(\set{\alpha}{\alpha_{t(\beta_0\cdots\beta_{i-1})}=\beta_i})
= 2^{-n}.
\end{align*}
Such processes are called \emph{tree processes}.

Tree processes are uniformly continuous in the standard metric: $T(\alpha)$ and $T(\beta)$ agree on their length-$n$ prefixes provided $\alpha$ and $\beta$ agree on their length-$m$ prefixes, where $m$ is the supremum of the labels on all nodes of depth $n$ or less in the tree $t$. 

\subsection{Scottifying the Cantor Space}
\label{sec:Scottify}

We can embed the Cantor space $\Bin$ in a Scott domain in a natural way. Consider the set $\Binp$ of finite and infinite binary strings ordered by the prefix relation. This is an algebraic CPO whose compact elements are the finite strings. Define $\up x = \set{y\in\Binp}{x\prec y}$ for $x\in\Binp$. The basic Scott-open sets are $\up x$ for $x\in 2^*$. 

\begin{lemma}\ 
	\label{lem:Scottify1}
	\begin{enumerate}[{\upshape(i)}]
		\item
		If $B$ is a Scott-open set of $\Binp$, then $B\cap\Bin$ is a Cantor-open set of $\Bin$.
		\item
		If $A$ is a Cantor-open set of $\Bin$, then $\set{x\in\Binp}{\up x\subs A}$ is a Scott-open set of $\Binp$, and is largest Scott-open set $B$ such that $B\cap\Bin=A$.
	\end{enumerate}
\end{lemma}
Thus the Cantor space $\Bin$ is a subspace of the Scott space $\Binp$.
The ``Scottified'' Cantor space gives an explicit characterization of functions that behave continuously with respect to coin usage in the sense that computations depend only on finite prefixes of the coin sequence.

Let $\D$ be a continuous $\omega$-CPO ordered by $\sqle$ with a meet operation $\bigsqcap$. Let $\prec$ be the proper prefix relation on strings.

\begin{lemma}\ 
\label{lem:Scottify2}
	\begin{enumerate}[{\upshape(i)}]
		\item
		If $f:\Binp\to\D$ is Scott-continuous, then $f\rest\Bin:\Bin\to\D$ is Cantor-continuous.
		\item
		If $g:\Bin\to\D$ is Cantor-continuous, then $g$ extends to a Scott-continuous map $\lamb x{\bigsqcap_{x\prec\alpha}g(\alpha)}:\Binp\to\D$.
	\end{enumerate}
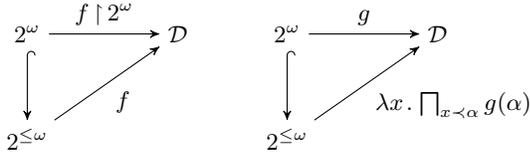
	\begin{align*}
	&
	\begin{tikzpicture}[->, >=stealth', node distance=20mm, auto]
	\small
	\node (NW) {$\Bin$};
	\node (NE) [right of=NW] {$\D$};
	\node (SW) [below of=NW, node distance=14mm] {$\Binp$};
	\path (NW) edge node {$f\rest\Bin$} (NE)
	edge[right hook->] node[swap] {} (SW);
	\path (SW) edge node[swap] {$f$} (NE);
	\end{tikzpicture}
	&&
	\begin{tikzpicture}[->, >=stealth', node distance=20mm, auto]
	\small
	\node (NW) {$\Bin$};
	\node (NE) [right of=NW] {$\D$};
	\node (SW) [below of=NW, node distance=14mm] {$\Binp$};
	\path (NW) edge node {$g$} (NE)
	edge[right hook->] node[swap] {} (SW);
	\path (SW) edge node[swap] {$\lamb x{\bigsqcap_{x\prec\alpha}g(\alpha)}$} (NE);
	\end{tikzpicture}
	\end{align*}
\end{lemma}
\begin{proof}
A proof can be found in the Appendi which can be found in the complete version of this paperx.
\end{proof}

\section{Stochastic Semantics}
\label{sec:stochastic}

We review briefly the stochastic semantics from \cite{BFKMPS18a}. This semantics was based on a map
\begin{align*}
\psem-:\mathsf{Exp}\to\mathsf{Env}\to\mathsf{Cont}\to\mathsf{Toss}\to\RV
\end{align*}
where
\begin{itemize}
	\item
	$\RV$ is the set of \emph{random variables} $\Omega\to\Val$ from a sample space $\Omega$ taking values in a reflexive CPO $\Val$,
	\item
	$\mathsf{Exp}$ is the set of \emph{stochastic $\lambda$-terms} $M$,
	\item
	$\mathsf{Env}$ is the set of \emph{environments} $E:\Var\to\RV$,
	\item
	$\mathsf{Cont}$ is the set of \emph{continuations} $C:\RV\to\RV$, and
	\item
	$\mathsf{Toss}$ is the set of \emph{tossing processes} $T:\Omega\to\Bin$.
\end{itemize}
Thus $\psem MECT:\RV$. The Boolean-valued semantics interpreted properties in the Boolean algebra of measurable sets of $\Omega$.

We can simplify the definition of \cite{BFKMPS18a} with a few observations.
\begin{enumerate}[(i)]
	\item
	In \cite{BFKMPS18a}, the map $\psem-$ is parameterized by continuations $C:(\Omega\to\Val)\to(\Omega\to\Val)$. Suppose we restrict continuations to be of the form $SF=\lamb f{\lamb\omega{F\omega(f\omega)}}$ for some $F:\Omega\to[\Val\to\Val]$, where $[\Val\to\Val]$ denotes the Scott-continuous deterministic maps.\footnote{The operation $S$ is the familiar $S$-combinator from combinatory logic.} Then all continuations that arise in the inductive definition of $\psem-$ are also of this form. Formally adopting this restriction allows us to eliminate continuations altogether, thereby simplifying the presentation. This also makes sense at an intuitive level: A single trial is a single evaluation of the program and depends only on one sample from $\Omega$.
	\item
	Tossing processes in \cite{BFKMPS18a} are of type $\Omega\to\Bin$, where $\Omega$ is an abstract sample space. A large part of the development of \cite{BFKMPS18a} was concerned with invariance properties of measure-preserving transformations of $\Omega$. For our purposes, there is no reason not to take the sample space to be $\Bin$ with the standard Lebesgue measure. Thus tossing processes become measure-preserving maps $T:\Bin\to\Bin$. This allows a more concrete treatment. A comprehensive characterization of such processes is given in \S\ref{sec:tossing}.
	\item
	The definition of \cite{BFKMPS18a} included a fixpoint operator. Our use of capsules allows us to eliminate this operator without loss of expressiveness.
\end{enumerate}

In addition to these simplifications, we introduce a more radical change that will admit a full-fledged operational semantics, namely the dynamic scoping of the random source.

The type of the semantic map is now
\begin{align*}
\psem{-}:\Exp\to\Env\to\Toss\to\RV.
\end{align*}
The values $\RV$ do not form a reflexive CPO, however they are built out of a reflexive CPO, as explained below in \S\ref{sec:FunLam}.
\begin{align*}
\Fun:\RV\to[\Toss\to\RV\to\RV]\\
\Lam:[\Toss\to\RV\to\RV]\to\RV
\end{align*}
We will define these functions explicitly below in \S\ref{sec:FunLam}.

\begin{definition}\ 
	\label{def:stochastic}
	\begin{enumerate}[(i)]
		\item
		$\psem{x}ET = E(x)$
		\item
		$\psem{MN}ET = \Fun(\psem{M}E(\proj_0\circ T))(\proj_1\circ T)(\psem{N}E(\proj_2\circ T))$
		\item
		$\psem{\lamb xM}ET = \Lam(\lamb{Tv}{\psem{M}E[v/x]T})$
		\item
		$\psem{M\oplus N}ET = \lamb\omega{\?{\hd(T\omega)}{\psem{M}E(\tl\circ T)\omega}{\psem{N}E(\tl\circ T)\omega}}$
	\end{enumerate}
	where clause (ii) uses the notation $\proj_i(\alpha)$ ($\alpha_i$ when evident from the context) to refer to the subsequence of $\alpha$ consisting of bits whose indices are $i\bmod 3$; thus $\proj_1(\alpha_0\alpha_1\alpha_2\cdots) = \alpha_1\alpha_4\alpha_7\cdots$, and clause (iv) uses the ternary predicate
	\begin{align}
	\?bst = \begin{cases}
	s, & \text{if $b=1$,}\\
	t, & \text{if $b=0$.}
	\end{cases}
	\label{eq:ternary}
	\end{align}
        
	In $\psem ME$, we assume that $\FV(M)\subs\dom E$.
\end{definition}

\section{Deterministic Semantics}
\label{sec:deterministic}

The observation of \S\ref{sec:stochastic} that allowed continuations to be eliminated can be carried further. All components in the definition of $\psem-$ are parameterized by sample points $\omega\in\Omega$, but as observed, there is no resampling in the course of a single trial; it is the same $\omega$. The function $\psem{-}$ does not really depend on the whole tossing process $T$ or the whole environment $E$, which are random variables, but only on their values. This observation allows us to develop an intermediate \emph{deterministic} denotational semantics in which all probabilistic choices are resolved in advance. The program runs deterministically, resolving probabilistic choices by consulting a preselected stack of random bits. In this section we introduce this semantics and develop some of its basic properties. Later, in \S\ref{sec:equivalence}, we will prove that it is equivalent to the stochastic semantics of \cite{BFKMPS18a} as modified in \S\ref{sec:stochastic} (Theorem \ref{thm:equivalence}). 

\subsection{A Domain of Values}

Barendregt \cite[\S5]{Ba84} presents several constructions of reflexive CPOs that can serve as denotational models of the untyped $\lambda$-calculus. One concrete such model, due to Engeler \cite{Engeler81,Longo83}, is a reflexive $\omega$-algebraic CPO $\powerset Q$ ordered by inclusion, where $Q$ is a certain countable set. The basic Scott-open sets are $\up a = \set b{a\subs b}$, where $a$ is a finite subset of $Q$. A function $\powerset Q\to\powerset Q$ is \emph{continuous} if it is continuous in this topology; equivalently, if $fb = \bigcup_{c\in\Pfin b}fc$.

In this section we present a version of the Engeler model modified to include a random source as an argument to continuous functions using the Scottified Cantor space of \S\ref{sec:Scottify}. Define
\begin{align*}
Q_0 &= \{\emptyset\} &
Q_{n+1} &= Q_n \uplus (2\star\times\Pfin{Q_n}\times Q_n) &
Q &= \bigcup_n Q_n
\end{align*}
and let $\Val = \powerset Q$, ordered by inclusion. The basic Scott-open sets of $\Val$ are $\up a = \set b{a\subs b}$, where $a\in\Pfin Q$. A function $\Val\to\Val$ is \emph{continuous} if it is continuous in this topology.

A function $f:\Bin\to\Val\to\Val$ is \emph{continuous} if it is continuous in both variables with respect to the Scott topology on $\Val$ and the Cantor topology on $\Bin$. The continuous functions of this type are denoted $[\Bin\to\Val\to\Val]$. Intuitively, $f$ is continuous if its value on $\alpha\in\Bin$ and $b\in\Val$ depends only on finite prefixes of $\alpha$ and finite subsets of $b$.

\begin{lemma}
	Let $f:[\Bin\to\Val\to\Val]$. Then
	\begin{align*}
	f\alpha b &= \bigcup_{c\in\Pfin b} \bigcup_{x\prec\alpha}\bigcap_{\beta\in I_x}f\beta c.
	\end{align*}
\end{lemma}
\begin{proof}
	Let $c\in\Val$. By the continuity of $f$ in its first argument, $\lamb\alpha{f\alpha c}:\Bin\to\Val$ is continuous, therefore for any basic open set $\up a$,
	\begin{align*}
	\alpha \in (\lamb\alpha{f\alpha c})^{-1}(\up a)
	&\Iff \exists x\prec\alpha\ I_x\subs (\lamb\alpha{f\alpha c})^{-1}(\up a).
	\end{align*}
	Then
	\begin{align*}
	a\subs f\alpha c
	&\Iff f\alpha c\in\up a\\
	&\Iff \alpha \in (\lamb\alpha{f\alpha c})^{-1}(\up a)\\
	&\Iff \exists x\prec\alpha\ I_x\subs (\lamb\alpha{f\alpha c})^{-1}(\up a)\\
	&\Iff \exists x\prec\alpha\ \forall\beta\in I_x\ \beta\in (\lamb\alpha{f\alpha c})^{-1}(\up a)\\
	&\Iff \exists x\prec\alpha\ \forall\beta\in I_x\ a\subs f\beta c\\
	&\Iff a\subs\bigcup_{x\prec\alpha}\bigcap_{\beta\in I_x}f\beta c.
	\end{align*}
	As $a$ was arbitrary, for any $b\in\Val$,
	\begin{align*}
	f\alpha b
	&= \bigcup_{c\in\Pfin b} f\alpha c
	= \bigcup_{c\in\Pfin b} \bigcup_{x\prec\alpha}\bigcap_{\beta\in I_x}f\beta c.
	\end{align*}
\end{proof}

To obtain a reflexive domain, we need to construct continuous maps
\begin{align*}
& \fun : \Val\to[\Bin\to\Val\to\Val]\\
& \lam : [\Bin\to\Val\to\Val]\to\Val
\end{align*}
such that $\fun\circ\lam=\id$.
\begin{align}
\fun a &= \lamb{\beta v}{\set{q\in Q}{\exists x\prec\alpha\ \exists b\in\Pfin v\ (x,b,q)\in a}}\nonumber\\
&= \bigcup_{x\prec\alpha}\bigcup_{c\in\Pfin b}\set{q}{(x,c,q)\in a}\label{eq:fundef}\\
\lam f &= \set{(x,c,q)\in 2^*\times\Pfin Q\times Q}{\forall\beta\in I_x\ q\in f\beta c}\nonumber\\
&\quad\cup\{\emptyset\}.\label{eq:lamdef}
\end{align}
Then
\begin{align*}
\fun(\lam f)\alpha b 
&= \bigcup_{x\prec\alpha}\bigcup_{c\in\Pfin b}\set{q}{(x,c,q)\in\lam f}\\
&= \bigcup_{x\prec\alpha}\bigcup_{c\in\Pfin b}\set{q}{\forall\beta\in I_x\ q\in f\beta c}\\
&= \bigcup_{c\in\Pfin b}\bigcup_{x\prec\alpha}\bigcap_{\beta\in I_x}f\beta c = f\alpha b.
\end{align*}
Also, note that since $\bot=\emptyset$ in $\Val$, $\fun\,\bot = \lamb{\beta v}\bot$, but $\lam\bot=\{\emptyset\}\ne\bot$. This is important for call-by-value, as we must distinguish $\Omega$ from $\lamb x\Omega$ for our adequacy result of \S\ref{sec:adequacy}.

\subsection{The Semantic Function}

For partial functions $f:D\pfun E$, define $\dom f = \set{x\in D}{\text{$f(x)$ is defined}}$. Equivalently, for functions $f:D\to E_\bot$, define $\dom f = \set{x\in D}{f(x)\ne\bot}$. Let $\FV(M)$ denote the free variables of $M$.

The type of our deterministic semantic function is
\begin{align*}
\dsem{-}:\Exp\to\Env'\to\Bin\to\Val,
\end{align*}
where $\Env' = \Var\to\Val$ is the set of (deterministic) environments.

\begin{definition}\ 
	\label{def:deterministic}
	\begin{enumerate}[(i)]
		\item
		$\dsem{x}e\alpha = e(x)$
		\item
		$\dsem{MN}e\alpha = \fun(\dsem{M}e(\proj_0(\alpha)))(\proj_1(\alpha))(\dsem{N}e(\proj_2(\alpha)))$
		\item
		$\dsem{\lamb xM}e\alpha = \lam(\lamb{\beta v}{\dsem{M}e\rebind vx\beta})$
		\item
		$\dsem{M\oplus N}e\alpha = \hd\alpha\,?\,\dsem{M}e(\tl\alpha) : \dsem{N}e(\tl\alpha)$
	\end{enumerate}
	where clause (iv) uses the ternary predicate \eqref{eq:ternary} and $\fun$ and $\lam$ are defined in \eqref{eq:fundef} and \eqref{eq:lamdef}, respectively. In $\dsem Me$, we assume that $\FV(M)\subs\dom e$.
	In (iii), we interpret the metaexpression $\lamb{\beta v}{\dsem{M}e\rebind vx\beta}$ as strict; thus
	\begin{align*}
	(\lamb{\beta v}{\dsem{M}e\rebind vx\beta})\,\alpha\,\bot=\bot.
	\end{align*}
\end{definition}
Note that this is completely deterministic. Probabilistic choices are resolved by consulting a preselected stack of random bits $\alpha$.

In the clause for $MN$, instead of $\evens$ and $\odds$ as in \cite{BFKMPS18a}, we divide the coins into three streams for use in, respectively, the evaluation of $M$, the evaluation of $N$,
and the application of the value of $M$ to the value of $N$.

This definition is well founded, but the resulting metaexpression is a $\lambda$-term that must be evaluated in the metasystem,
and that evaluation may not terminate. We define the value to be $\bot$ when that happens. For example,
consider $\dsem{\Omega}e\alpha$, where $\Omega=(\lamb x{xx})(\lamb x{xx})$. Define
\begin{align*}
u &\defeq \dsem{\lamb x{xx}}e\alpha\\
&= \lam\,(\lamb{\beta v}{\dsem{xx}e\subst vx\beta})\\
&= \lam\,(\lamb{\beta v}{\fun\,(\dsem xe\subst vx \prj 0\beta)\,\prj 1\beta\,(\dsem xe\subst vx \prj 2\beta)})\\
&= \lam\,(\lamb{\beta v}{\fun\,v\,\prj 1\beta\,v})\ne\bot.
\end{align*}
Note that this value is independent of $\alpha$, due to the fact that coins are dynamically scoped. Then
\begin{align*}
\dsem{\Omega}e\alpha
&= \dsem{(\lamb x{xx})(\lamb x{xx})}e\alpha\\
&= \fun(\dsem{\lamb x{xx}}e\prj 0\alpha)\,\prj 1\alpha\,(\dsem{\lamb x{xx}}e\prj 2\alpha)\\
&= \fun u\,\prj 1\alpha\,u\\
&= \fun (\lam(\lamb{\beta v}{\fun\,v\,\prj 1\beta\,v}))\,\prj 1\alpha\,u\\
&= (\lamb{\beta v}{\fun\,v\,\prj 1\beta\,v})\,\prj 1\alpha\,u\\
&= \fun\,u\,\prj{11}\alpha\,u\\
&= \cdots
\end{align*}

\begin{lemma}
	\label{lem:garbage}
	For $e_1,e_2:\Var\to\Val$, if $e_1\sqle e_2$ and $\FV(M)\subs\dom e_1$, then $\dsem Me_1\sqle\dsem Me_2$.
\end{lemma}
\begin{proof}
	Note that if $e_1\sqle e_2$, then $\dom e_1\subs\dom e_2$ and $e_1(x)\sqle e_2(x)$ for all $x\in\dom e_1$. 
	The proof is a straightforward induction on the structure of $M$.
	Suppose $e_1\sqle e_2$.
	\begin{align*}
	\dsem{x}e_1\alpha = e_1(x) \sqle e_2(x) = \dsem{x}e_2\alpha,\ x\in\dom e_1.
	\end{align*}
	\begin{align*}
	\dsem{MN}e_1\alpha
	&= \fun(\dsem{M}e_1(\proj_0(\alpha)))(\proj_1(\alpha))(\dsem{N}e_1(\proj_2(\alpha)))\\
	&\sqle \fun(\dsem{M}e_2(\proj_0(\alpha)))(\proj_1(\alpha))(\dsem{N}e_2(\proj_2(\alpha)))\\
	&= \dsem{MN}e_2\alpha.
	\end{align*}
	\begin{align*}
	\dsem{\lamb xM}e_1\alpha
	&= \lam(\lamb{\beta v}{\dsem{M}e_1\rebind vx\beta})\\
	&\sqle \lam(\lamb{\beta v}{\dsem{M}e_2\rebind vx\beta})\\
	&= \dsem{\lamb xM}e_2\alpha.
	\end{align*}
	\begin{align*}
	\dsem{M\oplus N}e_1\alpha
	&= \hd\alpha\,?\,\dsem{M}e_1(\tl\alpha) : \dsem{N}e_1(\tl\alpha)\\
	&\sqle \hd\alpha\,?\,\dsem{M}e_2(\tl\alpha) : \dsem{N}e_2(\tl\alpha)\\
	&= \dsem{M\oplus N}e_2\alpha.
	\end{align*}
\end{proof}

To extend $\dsem{-}$ to capsules, we combine a semantic environment $\Var\to\Val$ as used in Definition \ref{def:deterministic} and a capsule environment $\Var\to\Lambda_\bot$ in a single mixed environment $\sigma:\Var\to\Val+\Lambda_\bot$\footnote{In the coproduct, the $\bot$'s of the two domains are coalesced.}. From this we can obtain a new semantic environment $\sigma\star:\Var\to\Val$ as follows.
Consider the map
\begin{align*}
& P_\sigma:(\Var\to\Val)\to(\Var\to\Val)\\
& P_\sigma(\ell)(x) = \begin{cases}
\sigma(x), & \sigma(x)\in\Val\\
\dsem{\sigma(x)}\ell\alpha, & \sigma(x)\in\Lambda\\
\bot, & \sigma(x)=\bot
\end{cases}
\end{align*}
where in the second case we use Definition \ref{def:deterministic}(iii). The $\alpha$ there can be any element of $\Bin$, as $\dsem{\lamb xM}\ell:\Bin\to\Val$ is a constant function. The third case is already included in the first, so henceforth we omit explicit mention of it. Note that $\dom P_\sigma(\ell) = \dom\sigma$.

\begin{lemma}
	\label{lem:garbage0}
	$P_\sigma(\ell)$ is monotone in both $\sigma$ and $\ell$.
\end{lemma}
\begin{proof}
	Since $\Lambda_\bot$ is a flat domain, if $\sigma_1\sqle\sigma_2$ and $\sigma_1(x)\in\Lambda$, then $\sigma_1(x)=\sigma_2(x)$. It follows that for all $\ell:\Var\to\Val$ and $x\in\Var$,
	\begin{align}
	P_{\sigma_1}(\ell)(x)
	&= \begin{cases}
	\sigma_1(x), & \sigma_1(x)\in\Val\\
	\dsem{\sigma_1(x)}\ell\alpha, & \sigma_1(x)\in\Lambda
	\end{cases}\nonumber\\
	&\sqle \begin{cases}
	\sigma_2(x), & \sigma_2(x)\in\Val\\
	\dsem{\sigma_1(x)}\ell\alpha, & \sigma_2(x)\in\Lambda
	\end{cases}\label{eq:Psigma}\\
	&= \begin{cases}
	\sigma_2(x), & \sigma_2(x)\in\Val\\
	\dsem{\sigma_2(x)}\ell\alpha, & \sigma_2(x)\in\Lambda
	\end{cases}\nonumber\\
	&= P_{\sigma_2}(\ell)(x).\nonumber
	\end{align}
	If $\ell_1\sqle\ell_2$, then by Lemma \ref{lem:garbage},
	$\dsem{\sigma(x)}\ell_1\alpha \sqle \dsem{\sigma(x)}\ell_2\alpha$ for $\sigma(x)\in\Lambda$, therefore $P_{\sigma}(\ell_1)(x) \sqle P_{\sigma}(\ell_2)(x)$.
\end{proof}

By the Knaster-Tarski theorem, $P_\sigma$ has a least fixpoint
\begin{align}
\sigma\star(x) &= \begin{cases}
\sigma(x), & \sigma(x)\in\Val\\
\dsem{\sigma(x)}\sigma\star\alpha, & \sigma(x)\in\Lambda
\end{cases}
\label{eq:estar}
\end{align}
and we define $\dsem M\sigma\alpha=\dsem M\sigma\star\alpha$, where the right-hand side is by Definition \ref{def:deterministic}.
With this formalism, we have
\begin{align*}
\dsem{\rec fxM}\sigma &= \dsem{f}{\sigma\rebind{\lamb xM}f}.
\end{align*}

\begin{lemma}
	\label{lem:garbage1}
	If $\sigma_1,\sigma_2:\Var\to\Val+\Lambda_\bot$ are mixed environments with $\sigma_1\sqle\sigma_2$, then $\sigma_1\star\sqle\sigma_2\star$. In addition, if $\sigma_1(x)=\sigma_2(x)$ for all $x\in\dom\sigma_1$, then $\sigma_1\star(x)=\sigma_2\star(x)$ for all $x\in\dom\sigma_1\star=\dom\sigma_1$.
\end{lemma}
\begin{proof}
	From \eqref{eq:estar} and the fact that $\dsem{\lamb xM}e\alpha\ne\bot$ we have that $\dom\sigma_i = \dom \sigma_i\star$. By Lemma \ref{lem:garbage0}, we have $P_{\sigma_1}(\sigma_2\star) \sqle P_{\sigma_2}(\sigma_2\star) = \sigma_2\star$, so $\sigma_2\star$ is a prefixpoint of $P_{\sigma_1}$. Since $\sigma_1\star$ is the least prefixpoint, $\sigma_1\star\sqle\sigma_2\star$.
	
	In addition, if $\sigma_1$ and $\sigma_2$ agree on $\dom\sigma_1$, then for $x\in\dom\sigma_1$, equality holds in \eqref{eq:Psigma}, thus $P_{\sigma_1}(\ell)(x)=P_{\sigma_2}(\ell)(x)$. As this is true for all $\ell$, we have
	\begin{align*}
	\sigma_1\star(x) &= \sup_\alpha P_{\sigma_1}^\alpha(\bot)(x) = \sup_\alpha P_{\sigma_2}^\alpha(\bot)(x) = \sigma_2\star(x). &&
\qedhere
\end{align*}
\end{proof}

\begin{lemma}
	\label{lem:sigmastar}
	Let $\sigma:\Var\to\Val+\Lambda_\bot$ be a mixed environment.
	If $v\in\Lambda$, $y\not\in\FV(v)$, and $y\not\in\dom\sigma$, then
	$\sigma\rebind vy\star = \sigma\star\rebind{\dsem v\sigma\star\alpha}y$.
\end{lemma}
\begin{proof}
	\begin{align}
	\lefteqn{\sigma\rebind vy\star(y)}\nonumber\\
	&= P_{\sigma\rebind vy}(\sigma\rebind vy\star)(y)\nonumber\\
	&= \begin{cases}
	\sigma\rebind vy(y), & \sigma\rebind vy(y)\in\Val\\
	\dsem{\sigma\rebind vy(y)}\sigma\rebind vy\star\alpha, & \sigma\rebind vy(y)\in\Lambda
	\end{cases}\nonumber\\
	&= \dsem{v}\sigma\rebind vy\star\alpha\label{eq:sigmastarA}\\
	&= \dsem{v}\sigma\star\alpha\label{eq:sigmastarB}\\
	&= \sigma\star\rebind{\dsem v\sigma\star\alpha}y(y),\nonumber
	\end{align}
	where the inference \eqref{eq:sigmastarA} is because $\sigma\rebind vy(y)=v$ and $v\in\Lambda$ and the
	inference \eqref{eq:sigmastarB} is by Lemma \ref{lem:garbage1}. For $x\ne y$, $x\in\dom \sigma$,
	\begin{align}
	\lefteqn{\sigma\rebind vy\star(x)}\nonumber\\
	&= P_{\sigma\rebind vy}(\sigma\rebind vy\star)(x)\nonumber\\
	&= \begin{cases}
	\sigma\rebind vy(x), & \sigma\rebind vy(x)\in\Val\\
	\dsem{\sigma\rebind vy(x)}\sigma\rebind vy\star\alpha, & \sigma\rebind vy(x)\in\Lambda
	\end{cases}\nonumber\\
	&= \begin{cases}
	\sigma(x), & \sigma(x)\in\Val\\
	\dsem{\sigma(x)}\sigma\rebind vy\star\alpha, & \sigma(x)\in\Lambda
	\end{cases}\label{eq:sigmastarC}\\
	&= \begin{cases}
	\sigma(x), & \sigma(x)\in\Val\\
	\dsem{\sigma(x)}\sigma\star\alpha, & \sigma(x)\in\Lambda
	\end{cases}\label{eq:sigmastarD}\\
	&= P_\sigma(\sigma\star)(x)\nonumber\\
	&= \sigma\star(x)\nonumber\\
	&= \sigma\star\rebind{\dsem v\sigma\star\alpha}y(x),\nonumber
	\end{align}
	where the inference \eqref{eq:sigmastarC} is because $\sigma\rebind vy(x)=\sigma(y)$ and the
	inference \eqref{eq:sigmastarD} is by Lemma \ref{lem:garbage1}.
\end{proof}

\section{Relating the Stochastic and Deterministic Semantics}
\label{sec:equivalence}

In Definition \ref{def:stochastic}, we have reworked the semantic function of \cite{BFKMPS18a} to be of type
\begin{align*}
\psem{-}:\Exp\to\Env\to\Toss\to\RV,
\end{align*}
and in Definition \ref{def:deterministic} we have given a deterministic semantics of type
\begin{align*}
\dsem{-}:\Exp\to\Env'\to\Bin\to\Val,
\end{align*}
where
\begin{itemize}
	\item
	$\RV$ are the \emph{random variables} $\Omega\to\Val$ from a sample space $\Omega$ taking values in a reflexive CPO $\Val$,
	\item
	$\Exp$ are the \emph{stochastic $\lambda$-terms} $M$,
	\item
	$\Env$ are the (\emph{stochastic}) \emph{environments} $E:\Var\to\RV$,
	\item
	$\Env'$ are the (\emph{deterministic}) \emph{environments} $e:\Var\to\Val$,
	\item
	$\mathsf{Toss}$ are the \emph{tossing processes} $T:\Omega\to\Bin$.
\end{itemize}

In this section we establish the formal relationship between these two semantics (Theorem \ref{thm:equivalence}). The idea is that in the stochastic semantics, although all data are parameterized by a sample point $\omega\in\Omega$, it is actually the same $\omega$ throughout a single run of the program. All independence requirements are satisfied by the way randomness is allocated to the different tasks. For example, in the clause for $\psem{MN}$, the coin sequence is broken into three disjoint sequences to use in three distinct tasks (evaluation of $M$, evaluation of $N$, and application of $M$ to $N$). This is equivalent to three independent tossing processes. In the clause for $\oplus$, we resolve the probabilistic choice using the head coin, but then throw it away and continue with the tail of the coin sequence, so the head coin is not reused. Because of these considerations, linearity is maintained.

\subsection{$\Fun$ and $\Lam$}
\label{sec:FunLam}

We first show how to define $\Fun$ and $\Lam$ with the desired properties from $\fun$ and $\lam$.
Recall that
\begin{center}
	\begin{tikzpicture}[->, >=stealth', node distance=25mm, auto]
	\small
	\node (W) {$\Val$};
	\node (E) [right of=W] {$[\Bin\to\Val\to\Val]$};
	\path ([yshift=2pt]W.east) edge node[above] {$\fun$} ([yshift=2pt]E.west);
	\path ([yshift=-2pt]E.west) edge node[below] {$\lam$} ([yshift=-2pt]W.east);
	\end{tikzpicture}
\end{center}
and we need
\begin{center}
	\begin{tikzpicture}[->, >=stealth', node distance=45mm, auto]
	\small
	\node (W) {$\Omega\to\Val$};
	\node (E) [right of=W] {$[(\Omega\to\Bin)\to(\Omega\to\Val)\to(\Omega\to\Val)]$.};
	\path ([yshift=2pt]W.east) edge node[above] {$\Fun$} ([yshift=2pt]E.west);
	\path ([yshift=-2pt]E.west) edge node[below] {$\Lam$} ([yshift=-2pt]W.east);
	\end{tikzpicture}
\end{center}
We define $\Fun$ and $\Lam$ in two steps:
\begin{align*}
\Fun &= {\Fun}_2\circ{\Fun}_1 & \Lam &= {\Lam}_1\circ{\Lam}_2,
\end{align*}
where 
\begin{center}
	\begin{tikzpicture}[->, >=stealth', node distance=35mm, auto]
	\small
	\node (W) {$\Omega\to\Val$};
	\node (E) [right of=W] {$[\Omega\to(\Bin\to\Val\to\Val)]$};
	\path ([yshift=2pt]W.east) edge node[above] {${\Fun}_1$} ([yshift=2pt]E.west);
	\path ([yshift=-2pt]E.west) edge node[below] {${\Lam}_1$} ([yshift=-2pt]W.east);
	\node (W2) [below of=W, node distance=12mm] {\phantom{$\Omega\to\Val$}};
	\node (E2) [right of=W2, node distance=45mm] {$[(\Omega\to\Bin)\to(\Omega\to\Val)\to(\Omega\to\Val)]$};
	\path ([yshift=2pt]W2.east) edge node[above] {${\Fun}_2$} ([yshift=2pt]E2.west);
	\path ([yshift=-2pt]E2.west) edge node[below] {${\Lam}_2$} ([yshift=-2pt]W2.east);
	\end{tikzpicture}
\end{center}
$\Fun_1$ and $\Lam_1$ are just the covariant hom-functor in $\Set$ applied to $\fun$ and $\lam$, respectively.
\begin{align*}
& {\Fun}_1 = \Set(\Omega,\fun) = \fun\circ\,{-}\\
& {\Lam}_1 = \Set(\Omega,\lam) = \lam\circ\,{-}
\end{align*}
Then
\begin{align*}
({\Fun}_1\circ{\Lam}_1)f &= ((\fun\circ\,{-})\circ(\lam\circ\,{-}))f\\
&= (\fun\circ\,{-})(\lam\circ f) = \fun\circ\lam\circ f = f,
\end{align*}
so $\Fun_1\circ\Lam_1=\id$. Also,
\begin{align}
{\Lam}_1(\lamb\omega f) &= (\lam\circ\,-)(\lamb\omega f)\nonumber\\
&= \lam\circ\lamb\omega f
= \lamb\omega{(\lam\circ\lamb\omega f)\omega}\nonumber\\
&= \lamb\omega{\lam((\lamb\omega f)\omega)}
= \lamb\omega{\lam(f)}.\label{eq:Lam8}
\end{align}

We define
\begin{align*}
{\Fun}_2 &= \lamb{fg}{S(Sfg)},
\end{align*}
where $S=\lamb{gh\omega}{g\omega(h\omega)}$ is the familiar $S$-combinator from combinatory logic. Then
\begin{align}
{\Fun}_2fgh\omega &= S(Sfg)h\omega = Sfg\omega(h\omega) = (f\omega)(g\omega)(h\omega).\label{eq:Fun2}
\end{align}
The function ${\Fun}_2$ is injective:
\begin{align*}
\lefteqn{{\Fun}_2f_1 = {\Fun}_2f_2}\\
&\Imp \lamb{g}{S(Sf_1g)} = \lamb{g}{S(Sf_2g)}\\
&\Imp \forall g\forall h\forall\omega\ (f_1\omega)(g\omega)(h\omega) = (f_2\omega)(g\omega)(h\omega)\\
&\Imp \forall y\forall z\forall\omega\ f_1\omega yz = f_2\omega yz\\
&\Imp f_1 = f_2.
\end{align*}
We define $\Lam_2$ to be the inverse of $\Fun_2$ on the image of $\Fun_2$. Thus
\begin{align}
{\Lam}_2(\lamb{gh\omega}{(f\omega)(g\omega)(h\omega)}) &= {\Lam}_2({\Fun}_2f) = f.\label{eq:Lam9}
\end{align}
Also,
\begin{align*}
\lefteqn{{\Fun}_2({\Lam}_2(\lamb{gh\omega}{(f\omega)(g\omega)(h\omega)}))}\\
&= {\Fun}_2f
= \lamb{gh\omega}{(f\omega)(g\omega)(h\omega)},
\end{align*}
so ${\Fun}_2\circ{\Lam}_2 = \id$ on its domain. Then
\begin{align*}
\Fun\circ\Lam &= {\Fun}_2\circ {\Fun}_1\circ{\Lam}_1\circ {\Lam}_2\\
&= {\Fun}_2\circ {\Lam}_2 = \id.
\end{align*}
Moreover, using \eqref{eq:Lam8}, \eqref{eq:Fun2}, and \eqref{eq:Lam9},
\begin{align}
\lefteqn{\Fun fgh\omega = {\Fun}_2({\Fun}_1f)gh\omega}\nonumber\\
&= {\Fun}_2(\fun\circ f)gh\omega\nonumber\\
&= (\fun\circ f)\omega(g\omega)(h\omega) = \fun(f\omega)(g\omega)(h\omega)\label{eq:Fun5}\\
\lefteqn{\Lam(\lamb{gh\omega}{f(g\omega){(h\omega)}})}\nonumber\\
&= {\Lam}_1({\Lam}_2(\lamb{gh\omega}{(\lamb\omega f)\omega(g\omega){(h\omega)}}))\nonumber\\
&= {\Lam}_1(\lamb\omega f)
= \lamb\omega{\lam(f)}.\label{eq:Lam5}
\end{align}
Note that the domain $\Omega\to\Val$ is not reflexive with respect to $[(\Omega\to\Bin)\to(\Omega\to\Val)\to(\Omega\to\Val)]$ under $\Fun$ and $\Lam$, but only with respect to $[\Omega\to(\Bin\to\Val\to\Val)]$ (and its image in $[(\Omega\to\Bin)\to(\Omega\to\Val)\to(\Omega\to\Val)]$ under ${\Fun}_2$). However this is all we need for Theorem \ref{thm:equivalence}.

\subsection{Relating the Deterministic and Stochastic Semantics}

The following theorem gives the formal relationship between the stochastic and deterministic denotational semantics.
\begin{theorem}
	\label{thm:equivalence}
	$\lamb\omega{\dsem Me{(T\omega)}} = \psem M{(\lamb{x\omega}{ex})}T$.
\end{theorem}
\begin{proof}
  The proof follows by case analysis. A complete proof can be found in the Appendix which can be found in the complete version of this paper.

\end{proof}

\section{Operational Semantics}
\label{sec:opsem}

In this section we give big- and small-step operational rules in the style of \cite{Plotkin81} and prove their equivalence. The two styles use their coins in different patterns and the relationship must be formally specified. This is done using the tree processes of \S\ref{sec:trees}.

\subsection{Big-step rules}

The notation $\BOp Me\alpha vf$ means that $\cps Me$ reduces to $\cps vf$ under the big-step rules below with coins $\alpha\in\Bin$.
\begin{gather*}
\BOp xe\alpha{e(x)}e\\
\frac{\BOp Me{\alpha}vf}{\BOp{M\oplus N}e{0\alpha}vf}\qquad
\frac{\BOp Ne{\alpha}vf}{\BOp{M\oplus N}e{1\alpha}vf}\\[1ex]
\frac{\left\{\begin{array}c\BOp{M}e{\proj_0(\alpha)}{\lamb xK}{e_0}\\\BOp{N}{e_0}{\proj_1(\alpha)}{u}{e_1}\\\BOp {K\subst yx}{e_1\rebind uy}{\proj_2(\alpha)}vf\end{array}\right\}}{\BOp{MN}e\alpha vf}
\end{gather*}
where in the third premise of the last rule, $y$ is a fresh variable.

\subsection{Small-step rules}

The notation $\cps Me \red x \cps Nf$ means that $\cps Me$ reduces to $\cps Nf$ under the small-step rules below via a computation that consumes exactly coins $x\in 2^*$ in order from left to right.
The notation $\cps Me \red\alpha \cps Nf$ means that $\cps Me \red x \cps Nf$ for some $x\prec\alpha$, where $\alpha\in\Bin$.

\begin{gather*}
\frac{\cps Me \red x \cps{M'}f}{\cps{MN}e \red x \cps{M'N}f}
\qquad\frac{\cps Ne \red x \cps{N'}f}{\cps{vN}e \red x \cps{vN'}f}\\[1ex]
\cps Me \red\eps \cps Me
\qquad\cps xe \red\eps \cps{e(x)}e\\[1ex]
\cps{(\lamb xM)v}e \red\eps \cps{M\subst yx}{e \rebind vy} \quad \text{($y$ fresh)}\\[1ex]
\cps{M\oplus N}e \red 0 \cps Me 
\qquad\cps{M\oplus N}e \red 1 \cps Ne\\[1ex]
\frac{\cps Me \red x \cps Nf \quad \cps Nf \red y \cps Kg}{\cps Me \red{xy} \cps Kg}\\
\frac{\cps Me \red x \cps Nf\quad x\prec\alpha}{\cps Me \red\alpha \cps Nf}
\end{gather*}

\subsection{Relation of Big- and Small-Step Semantics}

The big- and small-step operational semantics use their coins in different patterns, and we need a way to characterize how they relate. The big-step rule for application breaks its coin sequence up into three independent coin sequences to evaluate the function, to evaluate the argument, and to apply the function, respectively;
whereas the small-step rules just use their coins sequentially.

The relationship is characterized by a \emph{tree process} as described in \S\ref{sec:tossing}.
The construction is given in the proof of the following theorem.

\begin{theorem}
	\label{thm:tree}
	For all $\cps Me$ there exists a tree process $T_{\cps Me}$ such that for all $\alpha$, $v$, $f$,
	\begin{align*}
	\BOp Me\alpha vf \Iff \cps Me \red{T_{\cps Me}(\alpha)} \cps vf.
	\end{align*}
\end{theorem}

\begin{proof}
 The rules for the big-step semantics define proof trees by which one concludes that
 an instance of the big-step relation holds. We proceed by induction on the
 structure of these proof trees. The base case corresponds to reading a
 variable from the environment: $\BOp xe\alpha{e(x)}e$. This case is immediate
 since we can just take the tree process to be the one that defines the
 identity function; note that the environment stores only values so there is no
 further reduction.

 For the case
 \[\frac{\BOp Me{\alpha}vf}{\BOp{M\oplus N}e{0\alpha}vf} \]
 we have, by induction, a tree process \(T_{\cps{M}e}\), call it $T'$ for
 short, such that
 \[ \cps Me \red{T'(\alpha)}\cps vf \]
 and analogously for the other branch of the choice.
 We can define the tree \(t_{\cps (M\oplus N)e}\) by
 \begin{align*}
  t_{\cps{M\oplus N}e}(\alpha) = &= \begin{cases}
   t_{\cps Me}(\alpha') \text{ if } \alpha = 0\alpha'\\
   t_{\cps Ne}(\alpha') \text{ if } \alpha = 1\alpha'
  \end{cases}
 \end{align*}
 It is a routine calculation to verify the result in this case.
 
	For the case $\cps{MN}e$, take the tree
	\begin{align*}
	t_{\cps{MN}e}(w) &= \begin{cases}
	3\cdot t_{\cps{K\subst yx}{e_1\rebind uy}}(z)+2,\\[2pt]
	\quad\parbox{5cm}{if $w=xyz\wedge\cps Me\red x\cps{\lamb xK}{e_0}$\\
		and $\cps N{e_0}\red y\cps u{e_1}$,}\\
	3\cdot t_{\cps N{e_0}}(y)+1,\\[2pt]
	\quad\parbox{5cm}{if $w=xy\wedge\cps Me\red x\cps{\lamb xK}{e_0}$\\
		and $\cps N{e_0}\red y\text{NV}$,}\\
	3\cdot t_{\cps Me}(w),\ \text{if $\cps Me\red w\text{NV}$},
	\end{cases}
	\end{align*}
	where NV means ``some capsule that is not reduced,''
	and let $T_{\cps{MN}e}$ be the associated tossing process.
	Then $\BOp{MN}e\alpha vf$ occurs iff there exist $K,u,e_0,e_1$, and $y$ fresh such that
	\begin{gather*}
	\BOp{M}e{\proj_0(\alpha)}{\lamb xK}{e_0}\qquad\BOp{N}{e_0}{\proj_1(\alpha)}{u}{e_1}\\
	\BOp {K\subst yx}{e_1\rebind uy}{\proj_2(\alpha)}vf
	\end{gather*}
	By the induction hypothesis, this occurs iff there exist $x,y,z$ such that
	\begin{align*}
	& \cps{M}e \red x \cps{\lamb xK}{e_0}
	&& x\prec T_{\cps Me}(\proj_0(\alpha))\\
	& \cps{N}{e_0} \red y \cps{u}{e_1}
	&& y\prec T_{\cps N{e_0}}(\proj_1(\alpha))\\
	& \cps{K\subst yx}{e_1\rebind uy} \red z \cps vf
	&& z\prec T_{\cps{K\subst yx}{e_1\rebind uy}}(\proj_2(\alpha))
	\end{align*}
	By construction of $t_{\cps{MN}e}(\alpha)$, $xyz \prec T_{\cps{MN}e}(\alpha)$, so this occurs iff
	\begin{align*}
	\cps{MN}e &\red{x} \cps{(\lamb xK)N}{e_0} \red{y} \cps{(\lamb xK)u}{e_1}\\
	&\red\eps \cps{K\subst yx}{e_1\rebind uy} \red{z} \cps vf,
	\end{align*}
	which occurs iff $\cps{MN}e \red {T_{\cps{MN}e}(\alpha)} \cps vf$.
\end{proof}

\section{Soundness and Adequacy}
\label{sec:adequacy}

%
%
%
%
%
%

The following theorem asserts the soundness and adequacy of our denotational semantics with respect to our big-step operational semantics.

\begin{theorem}\ 
	\label{thm:adequacy}
	\begin{enumerate}[{\upshape(i)}]
		\item
		If $\BOp M\sigma\alpha{\lamb xN}\tau$, then for any $\gamma$, $\dsem M\sigma\star\alpha = \dsem{\lamb xN}\tau\star\gamma = \lam\,(\lamb{\beta v}{\dsem N\tau\star\subst vx}\beta)$.
		\item
		If $\dsem M\sigma\star\alpha = \lam f$ for some $f:[\Bin\to\Val\to\Val]$, then $\converges M\sigma\alpha$.
	\end{enumerate}
\end{theorem}
\begin{proof}
	(i) The proof is by induction on the derivation of $\BOp M\sigma\alpha v\tau$.
	Let us do the easy cases first.
	For variables, we have $\BOp x\sigma\alpha{\sigma(x)}\sigma$ and
	\begin{align*}
	\dsem x\sigma\star\alpha = \sigma\star(x) = P_\sigma(\sigma\star)(x) = \dsem{\sigma(x)}\sigma\star\beta.
	\end{align*}
	
	For abstractions, we have $\BOp{\lamb xM}\sigma\alpha{\lamb xM}\sigma$ and
	\begin{align*}
	\dsem{\lamb xM}\sigma\star\alpha = \dsem{\lamb xM}\sigma\star\beta
	\end{align*}
	for any $\beta$, since the semantics of abstractions does not depend on $\alpha$.
	
	For choice, suppose $\BOp{M\oplus N}\sigma{\alpha}v\tau$. If $\hd\alpha = 0$, then $\BOp M\sigma{\tl\alpha}v\tau$. By the induction hypothesis, $\dsem M\sigma\star(\tl\alpha) = \dsem v\tau\star\beta$, so $\dsem{M\oplus N}\sigma\star(0\tl\alpha) = \dsem v\tau\star\beta$. By a similar argument, if $\hd\alpha = 1$, then $\dsem{M\oplus N}\sigma\star(1\tl\alpha) = \dsem v\tau\star\beta$. Thus in either case, $\dsem{M\oplus N}\sigma\star\alpha = \dsem v\tau\star\beta$.
	
	The most involved case is application. Suppose $\BOp{MN}\sigma{\alpha}v\tau$.
	Then for some $\lamb xK$, $\sigma_0$, $u$, and $\sigma_1$ such that $\sigma\sqle\sigma_0\sqle\sigma_1\sqle\tau$,
	\begin{gather*}
	\BOp{M}\sigma{\proj_0(\alpha)}{\lamb xK}{\sigma_0}
	\qquad\BOp{N}{\sigma_0}{\proj_2(\alpha)}{u}{\sigma_1}\\[1ex]
	\BOp {K\subst yx}{\sigma_1\rebind uy}{\proj_1(\alpha)}v\tau
	\end{gather*}
	where $y\in\Var$ is fresh.
	By the induction hypothesis,
	\begin{align}
	\dsem M\sigma\star(\proj_0(\alpha)) &= \dsem{\lamb xK}\sigma_0\star\beta\label{eq:IH1}\\
	\qquad \dsem N\sigma_0\star(\proj_2(\alpha)) &= \dsem u\sigma_1\star\beta\label{eq:IH2}\\
	\dsem{K\rebind yx}\sigma_1\rebind uy\star(\proj_1(\alpha)) &= \dsem v\tau\star\beta.\label{eq:IH3}
	\end{align}
	Then
	\begin{align}
	\lefteqn{\fun(\dsem M\sigma\star(\proj_0(\alpha)))}\nonumber\\
	&= \fun(\dsem{\lamb xK}\sigma_0\star\gamma) && \text{by \eqref{eq:IH1}}\nonumber\\
	&= \fun(\dsem{\lamb yK\rebind yx}\sigma_0\star\gamma),\ \text{$y$ fresh} && \text{by $\alpha$-conversion}\nonumber\\
	&= \fun(\lam(\lamb{\beta v}{\dsem{K\rebind yx}\sigma_0\star\rebind vy\beta}))\nonumber\\
	&= \lamb{\beta v}{\dsem{K\rebind yx}\sigma_0\star\rebind vy\beta}.\label{eq:K}
	\end{align}
	By \eqref{eq:IH2} and Lemma \ref{lem:garbage1}, since $\sigma_0$ is an extension of $\sigma$,
	\begin{align}
	\dsem N\sigma\star(\proj_2(\alpha)) = \dsem N\sigma_0\star(\proj_2(\alpha)) = \dsem u\sigma_1\star\beta.\label{eq:Nu}
	\end{align}
	By Lemma \ref{lem:sigmastar}, since $y$ is fresh,
	\begin{align}
	\sigma_1\rebind uy\star = \sigma_1\star\rebind{\dsem u\sigma_1\star\gamma}y.\label{eq:srebind}
	\end{align}
	Using \eqref{eq:IH3}--\eqref{eq:srebind} and Lemma \ref{lem:garbage1}, 
	\begin{align*}
	\lefteqn{\dsem{MN}\sigma\star\alpha}\\
	&= \fun(\dsem{M}\sigma\star(\proj_0(\alpha)))(\proj_1(\alpha))(\dsem{N}\sigma\star(\proj_2(\alpha)))\\
	&= (\lamb{\beta v}{\dsem{K\rebind yx}\sigma_0\rebind vy\beta})(\proj_1(\alpha))(\dsem u\sigma_1\star\gamma)\\
	&= \dsem{K\rebind yx}\sigma_0\star\rebind{\dsem u\sigma_1\star\gamma}y(\proj_1(\alpha))\\
	&= \dsem{K\rebind yx}\sigma_1\star\rebind{\dsem u\sigma_1\star\gamma}y(\proj_1(\alpha))\\
	&= \dsem{K\rebind yx}\sigma_1\rebind uy\star(\proj_1(\alpha))\\
	&= \dsem v\tau\star\beta.
	\end{align*}
	
	(ii) For variables, we have $\dsem x\sigma\star\alpha = \sigma\star(x) = \lam f$ for some $f:[\Bin\to\Val\to\Val]$. Since $\dom\sigma=\dom\sigma\star$, we must have $\sigma(x) = \lamb xK$ for some $\lamb xK\in\Lambda$. By definition of $\sigma\star$, $\sigma\star(x) = \dsem{\lamb xK}\sigma\star\alpha = \lam(\lamb{\beta v}{\dsem K\sigma\star\subst vx\beta})$. As $\lam$ is injective, $f=\lamb{\beta v}{\dsem K\sigma\star\subst vx\beta}$, and $\BOp x\sigma\alpha{\lamb yK}\sigma$.
	
	For $\lambda$-abstractions, $\dsem{\lamb xM}\sigma\star\alpha = \lam(\lamb{\beta v}{\dsem M\sigma\star\subst vx\beta})$ and $\BOp {\lamb xM}\sigma\alpha{\lamb xM}\sigma$.
	
	For choice, we have $\dsem{M\oplus N}\sigma\star\alpha = \hd\alpha\,?\,\dsem{M}\sigma\star(\tl\alpha) : \dsem{N}\sigma\star(\tl\alpha) = \lam f$. Either $\hd\alpha=1$, in which case $\dsem{M}\sigma\star(\tl\alpha) = \lam f$ and $\converges M\sigma{\tl\alpha}$ by the induction hypothesis, or $\hd\alpha=0$, in which case $\dsem{N}\sigma\star(\tl\alpha) = \lam f$ and $\converges N\sigma{\tl\alpha}$ by the induction hypothesis. In either case, $\converges{M\oplus N}\sigma\alpha$ by the big-step rule for choice.
	
	Finally, for applications, suppose $\dsem{MN}\sigma\star\alpha = \lam f$. We have
	\begin{align*}
	\dsem{MN}\sigma\star\alpha = \fun\,(\dsem{M}\sigma\star\prj 0\alpha)\,\prj 1\alpha\,(\dsem{N}\sigma\star\prj 2\alpha)
	\end{align*}
	If $\dsem{M}\sigma\star\prj 0\alpha = \bot$, then
	\begin{align*}
	\dsem{MN}\sigma\star\alpha
	&= \fun\,(\dsem{M}\sigma\star\prj 0\alpha)\,\prj 1\alpha\,(\dsem{N}\sigma\star\prj 2\alpha)\\
	&= \fun\,\bot\,\prj 1\alpha\,(\dsem{N}\sigma\star\prj 2\alpha)\\
	&= (\lamb{\beta v}\bot)\,\prj 1\alpha\,(\dsem{N}\sigma\star\prj 2\alpha)\\
	&= \bot,
	\end{align*}
	contradicting our assumption. Similarly, if $\dsem{M}\sigma\star\prj 0\alpha = \lam g$ but $\dsem{N}\sigma\star\prj 2\alpha = \bot$, then
	\begin{align*}
	\dsem{MN}\sigma\star\alpha
	&= \fun\,(\dsem{M}\sigma\star\prj 0\alpha)\,\prj 1\alpha\,(\dsem{N}\sigma\star\prj 2\alpha)\\
	&= \fun(\lam g)\,\prj 1\alpha\,\bot\\
	&= g\,\prj 1\alpha\,\bot\\
	&= \bot,
	\end{align*}
	again contradicting our assumption. So we can assume that $\dsem{M}\sigma\star\prj 0\alpha = \lam g$ and $\dsem{N}\sigma\star\prj 2\alpha \ne \bot$. By the induction hypothesis and Lemma \ref{lem:sigmastar},
	\begin{align*}
	& \BOp M\sigma{\alpha_0}{\lamb xK}{\sigma_0} && \BOp N{\sigma_0}{\alpha_2}{u}{\sigma_1}\\
	& g = \lamb{\beta v}{\dsem K\sigma_0\star\subst vx\beta} && \dsem N\sigma_0\star\alpha_2 = \dsem u\sigma_1\star(-).
	\end{align*}
	
	\begin{align*}
	\dsem{MN}\sigma\star\alpha
	&= \fun\,(\dsem{M}\sigma\star\prj 0\alpha)\,\prj 1\alpha\,(\dsem{N}\sigma_0\star\prj 2\alpha)\\
	&= \fun(\lam g)\,\prj 1\alpha\,(\dsem u\sigma_1\star(-))\\
	&= g\,\prj 1\alpha\,(\dsem u\sigma_1\star(-))\\
	&= (\lamb{\beta v}{\dsem K\sigma_0\star\subst vx\beta})\,\prj 1\alpha\,(\dsem u\sigma_1\star(-))\\
	&= \dsem K\sigma_0\star\subst{\dsem u\sigma_1\star(-)}x\prj 1\alpha\\
	&= \dsem{K\subst yx}\sigma_0\star\subst{\dsem u\sigma_1\star(-)}y\prj 1\alpha\\
	&= \dsem{K\subst yx}\sigma_1\star\subst{\dsem u\sigma_1\star(-)}y\prj 1\alpha\\
	&= \dsem{K\subst yx}\sigma_1\rebind uy\star\prj 1\alpha,
	\end{align*}
	and by the induction hypothesis, $\converges{K\subst yx}{\sigma_1\rebind uy}{\prj 1\alpha}$. By the big-step rule for application, $\converges{MN}\sigma\alpha$.
\end{proof}

\begin{corollary}
  For every capsule $\cps M \sigma$
  $$\{\alpha \in \Bin \, | \, \diverges M \sigma \alpha \} = \{\alpha \in \Bin \, | \, \dsem M \sigma\alpha = \bot \}$$

\end{corollary}

\section{Related Work and Concluding Remarks}
In probability theory, stochastic processes are modeled as random variables
(measurable functions) defined on a probability space, which is viewed as the
source of randomness. It is natural to think of probabilistic programming in a
similar vein, and in many of the related approaches one sees a programming
formalism augmented by a source of randomness.

The idea of modeling probabilistic programs with a stream of random data in the
$\lambda$-calculus has been used in \cite{borgstrom} as well. In that work, the
authors define big-step and small-step operational semantics for an idealized
version of the probabilistic language Church. Their operational semantics, like
ours, is a binary relation parameterized by a source of randomness. Although
their language can handle continuous distributions and soft conditioning, they
have not given a denotational semantics. Our approach could accommodate
continuous distributions simply by changing the source of randomness to
$\reals$. We might interpret $\reals$ either as the usual real numbers or as its
constructive version as in Real PCF \cite{realpcf}. To accommodate soft
conditioning, we could adopt the solution proposed in \cite{staton} of adding a
write-only state cell to store the weight of the execution trace, which can be
done by slightly changing our domain equation. These extensions would complicate
our semantics and its presentation, so we leave them for future work.

In a similar vein, the category of Quasi-Borel Spaces ($\mathbf{Qbs}$) defined
in \cite{qbs} also assumes that probability comes from an ambient source of
randomness, which they model as a set of random variables of type
$\mathbb{R} \to A$ satisfying certain properties. Furthermore, they show that in
$\mathbf{Qbs}$ there is a Giry-like monad that uses the set of random variables
in its definition. In \cite{wqbs}, to accommodate arbitrary recursion, they
equip every Quasi-Borel space with a complete partial order and require the set
of random variables to be closed with respect to directed suprema. It would be
interesting to better understand how our requirement of continuity of coin usage
relates to their construction.

There has also been alternative operational semantics for languages similar to
ours. In \cite{vignudelli}, an operational semantics is defined in terms of
Markov kernels over the values. Since the focus of that work is on syntactic
methods to reason about contextual equivalence, a denotational semantics is not
defined. However, by our adequacy and soundness theorems, we can also use our
semantics to reason about contextual equivalence. Furthermore, since the set
$\set{\alpha\in\Bin}{\converges M \sigma \alpha}$ is measurable, one can prove
by induction on reduction sequences that the semantics of \cite{vignudelli} and
ours are equivalent.

An alternative domain theoretical tool that has been used to interpret
randomness is the probabilistic powerdomain construction. Recently the Jung-Tix problem \cite{jungtix} has been solved \cite{jia2021}, showing that it is possible to define a commutative probabilistic monad in a cartesian closed category of continuous domains. We tackle the problem from a different perspective. We leave for future work to understand the connections between the probabilistic powerdomain and our functor $M X = \Binp \to X$. It is worth noting that this functor is the Reader monad from functional
programming. Unfortunately, if we were to use the same monad multiplication from
the Reader monad---i.e. $\mu_X (t) = \lambda \alpha. t \, \alpha \, \alpha$---we
would reuse the same source of randomness twice, breaking linearity of usage of
random data and probabilistic independence. Furthermore, for any other natural
transformation $M^2 \Rightarrow M$ that preserves such linearity conditions, the
associativity monad law holds only up to a measure-preserving function. As an
example, suppose that we chose the natural transformation
$\mu_X (t) = \lambda \alpha : \Bin. t \, \prj 0 \alpha \, \prj 1 \alpha $ as our
monad multiplication. In this case the associativity law becomes:
\begin{align*}
\lambda t : M^3(X) \, \alpha : \Bin.\, t \,\prj 0 {\prj 0 \alpha}\, \prj 0 {\prj 1 \alpha} \, \prj 1 \alpha = \\ \lambda t : M^3(X) \, \alpha : \Bin.\, t \, \prj 0 \alpha \, \prj 0 {\prj 1 \alpha} \, \prj 1 {\prj 1 \alpha}
\end{align*}
Obviously the equation above does not hold.

As a final example of a related formalism, we mention probabilistic coherence
spaces~\cite{Danos11,Ehrhard14}, which use the decomposition of the usual
function space into a linear function space and an exponential comonad.
In~\cite{Ehrhard14}, a fully abstract semantics is given for a probabilistic
extension to PCF. They model higher-order probability by using a generalization of transition matrices. Cones of
measures have also been used to construct a model of higher-order probabilistic
computation~\cite{Ehrhard17}. It is a fascinating question to understand
precisely the relationship between all these formalisms for higher-order
probabilistic computation.

To conclude, while other approaches to denotational semantics for higher-order probabilistic computation have been taken, no such construction is the obviously "correct" one. To clarify this matter, the connections between different approaches would need to be well understood. But this is a hard open problem and requires in-depth understanding of the various possible approaches and how they relate, and the field is not there yet. Even though the approaches mentioned above are interesting, we do not see that they have any compelling argument suggesting that they are the only "right" semantics for probabilistic higher-order computation. In this paper we contributed to the area by focusing on the Boolean-valued semantics of \cite{BFKMPS18a} and modified it to accommodate a call-by-value operational semantics which we proved it sound and adequate with respect to the modified denotational semantics, solving the main open problem from that work.


\section*{Acknowledgments}

Thanks to Giorgio Bacci, Fredrik Dahlqvist, Robert Furber and Arthur Azevedo de Amorim.
Special thanks to Dana Scott for many inspiring conversations.
Thanks to the Bellairs Research Institute of McGill University for providing a wonderful research environment.

This material is based upon work supported by a
grant from the National Science Foundation 
under grants 
No. AitF-1637532,
No. SaTC-1717581,
and
No. CCF-2008083.
Any opinions, findings, and conclusions or recommendations expressed
in this material are those of the authors and do not necessarily reflect
the views of the National Science Foundation.

Panangaden is funded by NSERC (Canada).


\bibliographystyle{IEEEtran}
\bibliography{mybib,prakash}
\appendix
\section{Appendix}

\begin{proof}[Proof of Lemma \ref{lem:nctp}]
Given any real $0<a\le 1$, we construct a dense open set $E_a$ such that $\Pr(E_a)=a$. Enumerate $x\in 2^*$ in order of increasing length, which enumerates the intervals $I_x = \set\alpha{x\prec\alpha}$ in order of decreasing size. Call an interval \emph{occupied} if it intersects $E_a$; initially, $E_a=\emptyset$ and all intervals are unoccupied.

Write the binary expansion of $a$ as $\sum_{i=0}^\infty 2^{-n_i}$, $n_0 < n_1 < \cdots$, using the form with trailing 1's when there are two representations. At stage $i$, let $I_x$ be the next unoccupied interval in the list. Set $E_a \defeq E_a\cup I_{xy}$, where $n_i=\len{xy}$. This is always possible, since the procedure maintains the invariant that the largest unoccupied interval is at least twice the size of the next occupying interval. This is true initially since the first occupying interval is of size at most $1/2$ and the first unoccupied interval is of size 1, and it is preserved in each step since populating an interval $I$ with an interval $I'$ at most half the size leaves an unoccupied subinterval of $I$ at least the same size $I'$, and the next occupying interval is at most half the size of $I'$. After each stage, remove all newly occupied intervals from the list and repeat. Every interval in the list eventually becomes occupied after all the intervals before it become occupied or are removed from the list.

This construction results in a dense open set $E_a$ such that $\Pr(E_a)=a$. The set $E_a$ is the union of countably many pairwise disjoint intervals $I_x$ for $x\in H_a$, where $H_a$ is some countable binary prefix code.

Now let $a_n$ be any sequence of reals $1/2<a_n<1$ such that $\prod_{n=0}^\infty a_n=1/2$, e.g.~$a_n=(2^{n+1}+1)/(2^{n+1}+2)$ or $(\binom{2n+3}2-1)/\binom{2n+3}2$, and let
\begin{align*}
A &= \prod_{n=0}^\infty H_{a_n} = \set{x_0x_1x_2\cdots}{x_n\in H_{a_n},\ n\ge 0}.
\end{align*}
The set $A$ is the intersection of a descending chain of dense open sets $A_m = \bigcup\,\set{I_x}{x\in\prod_{n=0}^m H_{a_n}}$, $m\geq 0$, with $\Pr(A_m)=\prod_{n=0}^m a_n$. Thus $A$ is a dense $G_\delta$ set with $\Pr(A)=1/2$.

We also claim that for all intervals $I$, $0<\Pr(A\cap I)<\Pr I$. Note that $1/2=\prod_{n=0}^\infty a_n<\prod_{n=1}^\infty a_n<1$, thus if $\Pr(A_0\cap I)=\eps$, then $\Pr(A\cap I)=\eps\prod_{n=1}^\infty a_n$, so $0<\eps/2\le\Pr(A\cap I)<\eps\le\Pr(I)$. Thus both $A$ and $\bar A=\Bin\setminus A$ are dense and of measure 1/2.

Since both $A$ and $\bar A$ are dense, there is no interval contained in either one of these sets. Moreover, no tossing process $T$ with $T^{-1}(\set\alpha{0\prec\alpha})=A$ can be made continuous by deleting a nullset $N$ from the domain, since then we must have $I\cap\bar N\subs A$ for some interval $I$, thus $I\cap\bar A\subs N$, contradicting the fact that $\Pr(\bar A\cap I)>0$ for all intervals $I$. Such a tossing process $T$ exists; set
\begin{align*}
T^{-1}(\set\alpha{0x\prec\alpha}) &= A\cap\set\alpha{\beta_a\lex\alpha\ltx\beta_b}\\ 
T^{-1}(\set\alpha{1x\prec\alpha}) &= \bar A\cap\set\alpha{\gamma_a\lex\alpha\ltx\gamma_b}
\end{align*}
for $a=.x000\cdots$ and $b=.x111\cdots$ in binary, where the $\beta_a,\gamma_a\in\Bin$ are chosen so that
\begin{align*}
\Pr(A\cap\set\alpha{\alpha\lex\beta_a}) &= \Pr(\bar A\cap\set\alpha{\alpha\ltx\gamma_a}) = a,
\end{align*}
and $\lex$ and $\ltx$ refer to lexicographic order on streams.
\end{proof}

\begin{proof}[Proof of Theorem \ref{thm:conttp}]
	If $T$ is a continuous partial or total tossing process, then $\set\alpha{x\prec T\alpha}$ is an open set, therefore can be written as a union of basic clopen sets $I_y$. Take $P_x$ be the set of $\prec$-minimal strings $y$ such that $I_y\subs\set\alpha{x\prec T\alpha}$. Intuitively, $P_x$ is the set of $\prec$-minimal prefixes of input streams that produce $x$ as a prefix of the output.
	This is a coding function, and $\Pr(\set\alpha{x\prec T\alpha})=2^{-\len x}$ by definition of tossing process. Moreover, if $T$ is total, then $x\mapsto P_x$ is exhaustive by uniform continuity.
	
	Conversely, every coding function satisfying (i)-(iii) gives rise to a continuous tossing process $T$ by defining $T\alpha$ to be the unique stream containing all prefixes $x$ such that $P_x\cap\set y{y\prec\alpha}\ne\emptyset$, or undefined if no such stream exists. This is a tossing process by Lemma \ref{lem:measurepreserving}.
\end{proof}

\begin{proof}[Proof of Lemma \ref{lem:Scottify2}]
	(i) By Lemma \ref{lem:Scottify1}(i), the inclusion map $\Bin\to\Binp$ is continuous, thus its composition with $f$ is.
	
	(ii) For any basic open set $\up x$ of $\D$,
	\begin{align*}
	& y\in(\lamb z{\bigsqcap_{z\prec\alpha}g(\alpha)})^{-1}(\up x)
	\Iff \bigsqcap_{y\prec\alpha}g(\alpha)\in \up x\\
	&\Iff x\sqle\bigsqcap_{y\prec\alpha}g(\alpha)
	\Iff \forall\alpha\ (y\prec\alpha \Imp x\sqle g(\alpha))\\
	&\Iff \forall\alpha\ (\alpha\in\up y \Imp \alpha\in g^{-1}(\up x))
	\Iff \up y \subs g^{-1}(\up x),
	\end{align*}
	and $\set y{\up y\subs g^{-1}(\up x)}$ is Scott-open by Lemma \ref{lem:Scottify1}(ii).
      \end{proof}

      \begin{proof}[Proof of Theorem \ref{thm:equivalence}]

	For variables $x$,
	\begin{align*}
	\lamb\omega{\dsem xe{(T\omega)}} = \lamb\omega{ex} = (\lamb{y\omega}{ey})x = \psem x{(\lamb{y\omega}{ey})}T.
	\end{align*}
	
	For choice,
	\begin{align}
	\lefteqn{\lamb\omega{\dsem{M\oplus N}e(T\omega)}}\nonumber\\
	&= \lamb\omega{\hd(T\omega)\,?\,\dsem{M}e(\tl(T\omega)) : \dsem{N}e(\tl(T\omega))}\label{eq:oplus2}\\
	&= \lamb\omega{\hd(T\omega)\,?\,\psem M(\lamb{x\omega}{ex})(\tl\circ T)\omega\nonumber\\
		&\qquad\qquad\qquad\ : \psem N(\lamb{x\omega}{ex})(\tl\circ T)\omega}\label{eq:oplus4}\\
	&= \psem{M\oplus N}(\lamb{x\omega}{ex})T.\label{eq:oplus6}
	\end{align}
	Step \eqref{eq:oplus2} is by definition of $\dsem-$ (Definition \ref{def:deterministic}).
	Step \eqref{eq:oplus4} is by the induction hypothesis for $M$ and $N$ and the fact that $\tl(T\omega) = (\tl\circ T)\omega$.
	Step \eqref{eq:oplus6} is by definition of $\psem-$ (Definition \ref{def:stochastic}).
	
	For application,
	\begin{align}
	\lefteqn{\lamb\omega{\dsem{MN}e(T\omega)}}\nonumber\\
	&= \lamb\omega{\fun(\dsem{M}e(\proj_0(T\omega)))(\dsem{N}e(\proj_1(T\omega)))(\proj_2(T\omega))}\label{eq:MN2}\\
	&= \lamb\omega{\fun(\psem{M}(\lamb{x\omega}{ex})(\proj_0\circ T)\omega)\nonumber\\
		&\qquad\quad(\psem{N}(\lamb{x\omega}{ex})(\proj_1\circ T)\omega)((\proj_2\circ T)\omega)}\label{eq:MN4}\\
	&= \lamb\omega{\Fun(\psem{M}(\lamb{x\omega}{ex})(\proj_0\circ T))\nonumber\\
		&\qquad\quad(\psem{N}(\lamb{x\omega}{ex})(\proj_1\circ T))(\proj_2\circ T)\omega}\label{eq:MN6}\\
	&= \Fun(\psem{M}(\lamb{x\omega}{ex})(\proj_0\circ T))\nonumber\\
	&\qquad\quad(\psem{N}(\lamb{x\omega}{ex})(\proj_1\circ T))(\proj_2\circ T)\label{eq:MN8}\\
	&= \psem{MN}T(\lamb{x\omega}{ex}).\label{eq:MN10}
	\end{align}
	Step \eqref{eq:MN2} is by definition of $\dsem-$.
	Step \eqref{eq:MN4} is by the induction hypothesis for $M$ and $N$ and the fact that $\proj_2(T\omega) = (\proj_2\circ T)\omega$.
	Step \eqref{eq:MN6} is by \eqref{eq:Fun5}.
	Step \eqref{eq:MN8} is by $\eta$-reduction.
	Step \eqref{eq:MN10} is by definition of $\psem-$.
	
	Finally, for $\lambda$-abstraction, we first need a property of rebinding:
	\begin{align}
	(\lamb{x\omega}{ex})[R/x]
	&= \lamb y{(y=x)\,?\,R : (\lamb{x\omega}{ex})y}\nonumber\\
	&= \lamb y{(y=x)\,?\,\lamb\omega{R\omega} : \lamb\omega{ey}}\nonumber\\
	&= \lamb{y\omega}{(y=x)\,?\,R\omega : ey}\nonumber\\
	&= \lamb{y\omega}{(e\rebind{R\omega}xy)}.\label{eq:bind5}
	\end{align}
	Then
	\begin{align}
	\lefteqn{\lamb\omega{\dsem{\lamb xM}e(T\omega)}}\nonumber\\
	&= \lamb\omega{\lam(\lamb{\beta v}{\dsem{M}(e\rebind vx)\beta})}\label{eq:lamb2}\\
	&= \Lam(\lamb{TR\omega}{\dsem M(e[R\omega/x]){(T\omega)}})\label{eq:lamb4}\\
	&= \Lam(\lamb{TR}{\psem{M}(\lamb{y\omega}{e[R\omega/x]y})T})\label{eq:lamb6}\\
	&= \Lam(\lamb{TR}{\psem{M}((\lamb{x\omega}{ex})[R/x])T})\label{eq:lamb8}\\
	&= \psem{\lamb xM}(\lamb{x\omega}{ex})T.\label{eq:lamb10}
	\end{align}
	Step \eqref{eq:lamb2} is by definition of $\dsem-$.
	Step \eqref{eq:lamb4} is by \eqref{eq:Lam5}.
	Step \eqref{eq:lamb6} is by the induction hypothesis for $M$.
	Step \eqref{eq:lamb8} is by \eqref{eq:bind5}.
	Step \eqref{eq:lamb10} is by definition of $\psem-$.

      \end{proof}

\end{document}